\documentclass[onecolumn,10pt]{IEEEtran}
\usepackage{amsfonts,amsmath,amssymb,amsthm,caption,algorithmic}
\usepackage{amsbsy}
\usepackage{graphicx,multirow,bm}
\usepackage{color}
\makeatletter
\newtheoremstyle{mythm}{3pt}{3pt}{}{16pt}{\bfseries}{:}{.5em}{}
\theoremstyle{mythm}
\newtheorem{theorem}{Theorem}
\newtheorem{example}{Example}
\newtheorem{definition}{Definition}
\newtheorem{remark}{Remark}

\newtheorem{corollary}{Corollary}
\newtheorem{lemma}{Lemma}

\newcommand{\INPUT}{\item[\algorithmicrequire]}
\renewcommand{\algorithmicrequire}{\textbf{Input:}}

\begin{document}
\title{
Optimal Locally Repairable Systematic Codes Based on Packings
\author{ Han Cai, Minquan Cheng, Cuiling Fan, and Xiaohu Tang, \IEEEmembership{Member,~IEEE}}
\thanks{H. Cai and X. H. Tang are with the Information Security and National Computing Grid Laboratory,
Southwest Jiaotong University, Chengdu, 610031, China (e-mail: hancai@aliyun.com, xhutang@swjtu.edu.cn).}
\thanks{M. Q. Cheng is with the Guangxi Key Lab of Multi-source Information Mining \& Security, Guangxi Normal University, Guilin 541004, China (E-mail:
chengqinshi@hotmail.com).
}
\thanks{C. L. Fan is with
the School of Mathematics, Southwest Jiaotong University,
Chengdu, 610031, China (e-mail: cuilingfan@163.com).}
}
\date{}
\maketitle

\vspace{0.1in}
\begin{abstract}
Locally repairable codes are desirable for distributed storage systems to improve the repair efficiency. In this paper, we first build a bridge between  locally repairable code and  packing. As an application of this bridge, some optimal locally repairable codes can be obtained by packings, which gives optimal locally repairable codes with flexible parameters.
\end{abstract}


\section{Introduction}
Right now, large-scale cloud storage and distributed file systems
such as Amazon Elastic Block Store (EBS) and
Google File System (GoogleFS) have reached such a massive scale that the disk failures are
the norm and not the exception. In these systems, to protect the data from disk failures, the simplest solution is the
straightforward replication of data packets across different disks. In addition to fault tolerance,
replication has good parallel reading ability \cite{RPDV} that a content requested by multiple users can be directed to
different replicas, which are very important for those hot data  needed to be read frequently.
However, unfortunately this solution suffers larger storage overhead. Accordingly, an alternative solution
based on storage codes was proposed.

In storage system, an $[n,k]$ storage code encodes $k$ information symbols to $n$ symbols and stores them
across $n$ disks. Generally speaking, among all the storage codes, maximum distance separable (MDS) code
is preferred for the practical systems because
it can lead to dramatic improvements both in terms of redundancy and reliability compared with replication \cite{GHSY}.
Nevertheless, an $[n,k]$ MDS code has a drawback that whenever recovering a symbol one needs to connect $k$ surviving symbols. This is expensive especially in large-scale distributed file systems. To overcome this drawback,  locally repairable code was introduced to reduce the number of symbols connected during the repair process \cite{HCL}.

The concept of locally repairable codes was initially studied in \cite{HCL}, where a symbol can be recovered from accessing only other
$r\ll k$ symbols. Later, this concept was generalized to the case that even if multiple disk failures occur, the failed node can still be recovered locally \cite{PKLK}. In 2014, a new kind of locally repairable codes was proposed by Wang \textit{et al.} from a combinatorial perspective, which also has the ability to recover multiple disk failures locally \cite{WZ}. Specifically, this code has a  property that it can ensure every information symbol with disjoint local repair groups, each of which can be used to reconstruct the target information symbol locally. Consequently, it has the advantage of good parallel reading ability since each repair group can be seen as a backup for the  target information symbol and then
can be accessed independently \cite{RPDV}. In addition,  the locally repairable code with multiple repair groups can have higher code rate  in some special cases  in contrast to  the one in \cite{PKLK}. Meanwhile, some upper bounds on the minimum Hamming distance of locally repairable codes were derived, such as the Singleton-type bound in \cite{GHSY,PD,PKLK}, the bound depending on the size for alphabet \cite{CM}, the bound for locally repairable codes with multiple erasure tolerance \cite{RPDV,WZ}, etc. Up to now, Numerous constructions of optimal locally repairable codes with respect to those bounds have been reported in the literature,  e.g., see \cite{CHL,FY,GHJY,GHSY,HCL,PHO,PD,PKLK,RKSV,RPDV,SDYL,TPD,TB,WZ}, and the references therein.

Very recently, Rawat \textit{et al.} generalized  the locally repairable code with each information symbol having multiple disjoint repair groups to the nonlinear case \cite{RPDV}.
In particular, in \cite{RPDV}, Rawat \textit{et al.} { derived} an upper bound on the minimum Hamming distance of a specific class of such codes, in which each repair group contains exactly one check symbol. So far, there are only two optimal constructions with respect to this bound: One is based on resolvable designs \cite{RPDV}; Another is via partial geometry \cite{PHO}. However, the constraints for both the resolvable design and the partial geometry are so strong that only a few results are known.

Besides, in storage system the data itself may change frequently, especially for the hot data, which requires us to improve the
update-efficiency \cite{MCW}. That is, the optimal update-efficiency is also very desirable in the practical systems.
Therefore, in this paper we focus on the locally repairable codes with optimal update-efficiency and each information symbol having multiple disjoint repair groups, where each repair group contains exactly one check symbol. Firstly, we combinatorially characterize  optimal locally repairable codes with respect to the bound in \cite{RPDV} via packing, which is a simple well studied combinatorial structure (e.g., see \cite{CFM,CK,CY,CSW,MOKL,Y,YonJCTA}). Secondly some general constructions of optimal locally repairable codes with
optimal update-efficiency are presented based on packings.  In particular, sufficient and necessary conditions for optimal locally repairable code are obtained for some special cases.

The remainder of this paper is organized as follows. Section \ref{sec-preliminaries}
introduces some preliminaries about locally repairable codes.
Section \ref{sec-relationship} { proposes a combinatorial characterization about locally repairable codes.}
Sections \ref{sec-construction}  presents general constructions of locally repairable codes { based on packings}. In the meantime  some packings that can be used to generate optimal locally repairable codes are proposed.
Section \ref{sec-conclusion} concludes
this paper with some remarks.

\section{Preliminaries}\label{sec-preliminaries}

Throughout this paper, we use the following notations:
\begin{itemize}
\item For a positive integer $n$, let  $[n]$ denote the set $\{1,2,\cdots,n\}$;
\item For any prime power $q$, let $\mathbb{F}_q$ denote the finite field with $q$ elements;
\item Let $x = (x_1, \cdots , x_n)$ be a vector, and $supp(x) = \{i | x_i\ne 0, 1 \le i \le n\}$ denote its support;
\item An $[n,k]$ linear code $\mathcal{C}$ over $\mathbb{F}_q$  is a $k$-dimensional subspace of $\mathbb{F}_q^n$
yielded by a $k\times n$ generator matrix
$G=({\bf g}_1,{\bf g}_2,\ldots,{\bf g}_{n})$, where ${\bf g}_i$ is a column vector of dimension $k$ for all $1\le i\le n$. Specifically,
it is said to be an $[n,k,d]$ linear code if the minimum Hamming distance is $d$;
\item For a subset $S\subseteq [n]$, let $span(S)$ be the linear space
spanned by $\{{\bf g}_i |i\in  S\}$ over $\mathbb{F}_q$ and $rank(S)$ be the dimension of $span(S)$.
\end{itemize}

\subsection{Locally Repairable Codes}

The $i$th $(1 \leq i \leq n)$ code symbol $c_i$ of an $[n, k,d]_q$ linear code $\mathcal{C}$
is said to have locality   $r$ $(1 \leq r \leq k)$, if it can be recovered
by accessing at most $r$ other symbols in $\mathcal{C}$. More precisely,
symbol locality can also be defined in mathematical way as follows.

\begin{definition}[\cite{HCL}]\label{def_c_local}
For any column ${\bf g}_i$ of $G$ with $i\in [n]$, define Loc$({\bf g}_i)$ as the smallest integer $r$ such that there exists
$r$ integers $i_1,i_2,\cdots,i_r\in [n]\backslash\{i\}$ satisfying
$${\bf g}_i=\sum_{t=1}^{r}\lambda_t{\bf g}_{i_t}\ \ \ \ \lambda_t\in \mathbb{F}_q$$
and define Loc$(S)=\max\limits_{ i\in S}{\rm Loc}({\bf g}_i)$ for any set $S\subseteq[n]$.
Then, an $[n,k,d]_q$ linear code $\mathcal{C}$ is said to have information locality $r$ if there exists $S\subseteq [n]$ with $rank(S)=k$
satisfying ${\rm Loc}(S)\leq r.$
\end{definition}

\subsection{Locally repairable codes for multiple disk failures}

\begin{definition}[\cite{PKLK}]
The $i$th code symbol $c_i$, $1\leq i\leq n$, in an
$[n, k, d]_q$ linear code $\mathcal{C}$, is said to have $(r, \delta)_i$-locality  if
there exists a subset $S_i\subseteq
[n]$ such that
\begin{itemize}
  \item $i\in S_i$ and $|S_i|\leq r+\delta-1$; and
  \item the minimum distance of the code $\mathcal{C}|_{S_i}$ obtained by deleting code symbols $c_i$ ($i \in [n]\setminus S_i$) is at least $\delta$.
\end{itemize}
Further, an $[n,k,d]_q$ linear code $\mathcal{C}$ is said to have information $(r,\delta)_i$-locality if there exists $S\subseteq [n]$ with $rank(S)=k$
such that for each $i\in S$, the $i$th code symbol has $(r, \delta)_i$-locality.
\end{definition}

\begin{lemma}[\cite{PKLK}] \label{lemma_bound_i}
The minimum distance $d$ of a code $\mathcal{C}$ with $(r,\delta)_i$ locality is upper bounded by
\begin{equation}
d\leq n-k+1-\left(\left\lceil\frac{k}{r}\right\rceil-1\right)(\delta-1)
\end{equation}
\end{lemma}

\subsection{Locally repairable codes for multiple disk failures with good parallel reading ability}

\begin{definition}[\cite{WZ}]\label{def_r_delta_c}
The $i$th code symbol $c_i$, $1\leq i\leq n$, of an $[n,k,d]$ linear code $\mathcal{C}$ is said to have $(r,\delta)_c$-locality if
there exist $\delta-1$ pairwise disjoint sets $R^{(i)}_1,R^{(i)}_2,\cdots,R^{(i)}_{\delta-1}\subseteq [n]\backslash\{i\}$, satisfying
\begin{itemize}
  \item $\left|R^{(i)}_j\right|\leq r$; and
  \item ${\bf g}_i\in span\left(R^{(i)}_j\right)$
\end{itemize}
for all $0\le j<\delta$ where each $R^{(i)}_j$ is called a repair group of ${\bf g}_i$. Further,
 a code $\mathcal{C}$ is said to have information $(r,\delta)_c$-locality if there is a
subset $S\subseteq [n]$ with $rank(S)=k$ such that for each $i\in S$, the $i$th code symbol
has $(r, \delta)_c$-locality.
\end{definition}

\begin{lemma} [\cite{RPDV}]\label{new_bound}
For an $[n,k,d]_q$ linear code with information $(r,\delta)_c$-locality, then
\begin{equation}\label{eqn_bound_with_condition}
d\leq n-k-\left\lceil\frac{k(\delta-1)}{r}\right\rceil+\delta
\end{equation}
if there is only one check symbol in each repair group.
\end{lemma}

\begin{remark} For the case { $r|k$, the} above bound is exactly the one in Lemma \ref{lemma_bound_i}. While for
the case $r\nmid k$, $\left\lceil k(\delta-1)/ r\right\rceil\leq \left\lceil\frac{k}{r}\right\rceil(\delta-1)$ implies
that $$k+d+\left\lceil\frac{k(\delta-1)}{r}\right\rceil-\delta\leq k+d-1+\left(\left\lceil\frac{k}{r}\right\rceil-1\right)(\delta-1)$$
This is to say, compared with the optimal codes with information $(r,\delta)_i$-locality, there may exist shorter optimal codes with information
$(r,\delta)_c$-locality
for the case $r\nmid k$.  The following example shows that such code indeed exists.
\begin{example}
\label{exa1}
For the case $k=8$, $r=\delta-1=3$, let
\begin{equation}\label{eqn_matrix}
\small
G=\left(
\begin{array}{rcccccccccccccccl}
1&0&0&0&0&0&0&0&0&1&0&0&0&0&1&1\\
0&1&0&0&0&0&0&0&1&0&1&0&0&0&0&1\\
0&0&1&0&0&0&0&0&1&1&0&1&0&0&0&0\\
0&0&0&1&0&0&0&0&0&1&1&0&1&0&0&0\\
0&0&0&0&1&0&0&0&0&0&1&1&0&1&0&0\\
0&0&0&0&0&1&0&0&0&0&0&1&1&0&1&0\\
0&0&0&0&0&0&1&0&0&0&0&0&1&1&0&1\\
0&0&0&0&0&0&0&1&1&0&0&0&0&1&1&0\\
\end{array}
\right)
\end{equation}
It is easy to check that the linear code $\mathcal{C}$  generated by $G$  has information $(3,4)_c$-locality and minimum Hamming distance $d=4$. By lemma \ref{new_bound}, $\mathcal{C}$ is an optimal  $[16,8,4]$ linear code. But by lemma \ref{lemma_bound_i}, the optimal linear code with information $(3,4)_i$-locality has length $n\geq k+d-1+\left(\left\lceil\frac{k}{r}\right\rceil-1\right)(\delta-1)=17>16.$
\end{example}
\end{remark}

\subsection{The update-efficiency of codes}

\begin{definition}[\cite{MCW}]\label{Def_Update}
The \textit{update-efficiency} of the code $\mathcal{C}$ is the maximum number of symbols that need to be changed when there is one symbol changed for the message.
\end{definition}

\begin{lemma}[\cite{MCW}]\label{lemma_bound_updata}
For any binary $[n,k,d]$ linear  code $\mathcal{C}$, its update-efficiency $t$ satisfies $t\ge d$.
\end{lemma}

Although the above result was proved only for binary code in \cite{MCW}, it is easy to check that the inequality
$t\geq d$ also holds for the non-binary case. In general, the update-efficiency should be as small as possible.  Nevertheless,
the lower bound $t\geq d$ tell us that the minimum update-efficiency of an $[n,k,d]$ linear code $\mathcal{C}$ is at least $d$.
Hence,

\begin{definition}\label{Def_Update_New}
An $[n,k,d]$ linear  code $\mathcal{C}$ is said to have optimal update-efficiency if its update-efficiency is $d$.
\end{definition}

\subsection{Packing}

Finally, we review packing, the main combinatorial tool used in this paper.

\begin{definition}[\cite{CD}]\label{Def_Pack}
Let $R$ be a subset of positive integers and $k\geq 2$ be an integer.
A $(k,R,1)$ packing is a two tuple $(X,\mathcal{B})$ where $X$ is a set of $k$ elements and
$\mathcal{B}$ is a collection of subset of $X$ called blocks that satisfies
\begin{itemize}
\item $R=\{|B|: B \in \mathcal{B}\}$;
\item every pair of distinct elements of $X$ occurs in at most one block of $\mathcal{B}$.
\end{itemize}
\end{definition}

If $R=\{r\}$, a $(k,R,1)$ packing is also denoted as $(k,r,1)$ \textit{packing}. Moreover, a $(k,R,1)$ packing is said
to be \textit{regular} if each element of $X$ appears in exactly $t$ blocks, denoted by $t$-regular $(k,R,1)$ packing.

\begin{definition}[\cite{CD}]\label{Def_Pack-Res}
A packing $(X,\mathcal{B})$, denoted by $(k,R,1;u)$ packing, is said to be \textit{resolvable} if
\begin{itemize}
\item $\mathcal{B}=\bigcup\limits_{i=1}^{u}\mathcal{B}_i$ with $\mathcal{B}_i\bigcap\mathcal{B}_j=\emptyset$ for any $i\ne j\in[u]$;
\item For any $i\in [u]$, $\mathcal{B}_i$ is a partition of $X$, i.e., $X=\bigcup\limits_{B\in \mathcal{B}_i} B$ and $B\bigcap B'=\emptyset$ for any $B\ne B'\in \mathcal{B}_i$.
\end{itemize}
\end{definition}

Obviously, a $(k,R,1;u)$ packing is an $u$-regular packing.

\begin{example}
The two tuple $(X, \mathcal{B})$ with $X=[8]$ and $\mathcal{B}=\mathcal{B}_1\bigcup \mathcal{B}_2\bigcup \mathcal{B}_3\bigcup\mathcal{B}_4$ is a $(8,\{3,2\},1;4)$ resolvable packing, where
\begin{equation*}
\begin{split}
\mathcal{B}_1=\{\{2,3,8\},\,\{6,7,4\},\,\{1,5\}\},\ \ \mathcal{B}_2=\{\{3,4,1\},\,\{7,8,5\},\,\{2,6\}\}\\
\mathcal{B}_3=\{\{4,5,2\},\,\{8,1,6\},\,\{3,7\}\},\ \ \mathcal{B}_4=\{\{5,6,3\},\,\{1,2,7\},\,\{4,8\}\}
\end{split}
\end{equation*}
\end{example}

\section{The combinatorial characterization of Locally repairable codes via packing}\label{sec-relationship}

For simplicity, from now on we always assume that the generator matrix $G$ of $\mathcal{C}$ is of the canonical form.
That is,
\begin{eqnarray}\label{Eqn_Gen_Matrix}
G=({\bf e}_1,{\bf e}_2,\ldots,{\bf e}_{k}\ |\ {\bf p}_1,\ldots,{\bf p}_{n-k})=(I_k\ |\ P)
\end{eqnarray}
where $I_k=({\bf e}_1,{\bf e}_2,\ldots,{\bf e}_{k})$ is the $k\times k$ identity matrix, ${\bf e}_i$ and ${\bf p}_j$ are column vectors of length $k$ for $1\le i\le k$
and $1\le j\le n-k$.
It is well known that the resultant code $\mathcal{C}$ is a systematic code whose information symbols $c_1,\cdots, c_k$ and check symbols $c_{k+1},\cdots,c_n$
correspond to the columns ${\bf e}_1,{\bf e}_2,\ldots,{\bf e}_{k}$ and  ${\bf p}_1,\ldots,{\bf p}_{n-k}$ respectively. So, in this paper
we call ${\bf p}_1,\ldots,{\bf p}_{n-k}$ the check columns.

Given an information symbol $c_i$, consider its repair group $R_j^{(i)}$, where $1\le i\le k$ and $1\le j\le \delta-1$.
Note that there is only one check symbol $c_{k+l}$ in $R_j^{(i)}$ for $1\le l\le n-k$. Clearly, $supp({\bf p}_l)= \{i\}\cup \left(R_j^{(i)}\setminus\{k+l\}\right)$.
That is,  a repair group is completely determined by a check column, and vice versa. Therefore, we have  the following alternative
definitions of a systematic code with information $(r,\delta)_c$-locality and its update-efficiency to Definitions \ref{def_r_delta_c} and \ref{Def_Update}.

\begin{definition}\label{r_delta_i_one}
For any integer $1\leq i\leq k$, if there exist distinct $\delta-1$ columns ${\bf p}^{(i)}_{1}$, ${\bf p}^{(i)}_{2}$, $\ldots$, ${\bf p}^{(i)}_{\delta-1}$
of a $k\times (n-k)$ matrix $P$ satisfying
\begin{itemize}
  \item $\left|supp\left({\bf p}^{(i)}_{j}\right)\right|\leq r$ with $1\leq j\leq \delta -1$; and
  \item $\{i\}= supp\left({\bf p}^{(i)}_{j}\right)\bigcap supp\left({\bf p}^{(i)}_{t}\right)$ for any two integers $1\leq j\neq t\leq \delta -1$,
\end{itemize}
then the systematic code $\mathcal{C}$ generated by $G=(I_k\ |\ P)$ is said to have information $(r,\delta)_c$-locality. Further, $\mathcal{C}$ is said to be
an optimal $[n,k,d]$  systematic code with information $(r,\delta)_c$-locality if it achieves the bound in \eqref{eqn_bound_with_condition}.
\end{definition}

\begin{definition}\label{Def_Update_New_1}
The \textit{update-efficiency} of the systematic code $\mathcal{C}$ generated by the matrix $G$ in \eqref{Eqn_Gen_Matrix} is the maximum Hamming weight of the rows of $G$.
\end{definition}

As for the aforementioned repair group $R_j^{(i)}$ , recall that $\left|R_j^{(i)}\right|\le r$ , so does $|supp({\bf p}_l)|$. Consequently,  we divide the check symbols $c_{k+l}$ ($1\le l\le n-k$) into two subsets according to the Hamming weight of the corresponding column ${\bf p}_{l}$: $c_{k+l}$ and ${\bf p}_{l}$ are said to be partial check symbol and partial check column respectively if $\left|supp\left({\bf p}_{l}\right)\right|\leq r$; Otherwise they are said to be non-partial check symbol and non-partial check column. Denote the number of partial check symbols and non-partial check symbols by $n_1$ and $n_2$ respectively. It is clear that $n_1+n_2+k=n$. Without loss of generality (W.L.O.G.), assume that $|supp({\bf p}_{i})|\le r$ if and only if $1\le i\le n_1$.

Obviously, $supp({\bf p}_i)$ ($1\le i\le n_1$) is crucial to study the  locality property.
In this section, for an optimal $[n,k,d]$  systematic code $\mathcal{C}$ with information $(r,\delta)_c$-locality,  we
investigate the combinatorial structure of the supports of partial check columns. To this end, we first characterize the supports of an $[n,k]$ systematic code with information $(r,\delta)_c$-locality  via packing.

\begin{lemma}\label{lem_comb}
Given an $[n,k]$ systematic code $\mathcal{C}$ with information $(r,\delta)_c$-locality. For any two distinct elements, $i_1$, $i_2\in[ k]$,  if the pair  $(i_1, i_2)$ occurs in $t> 1$ support sets of partial check columns, then  $i_1$ and $i_2$ occur in no less than $\delta+t-2$ support sets of  partial check columns respectively;
\end{lemma}

\begin{proof}
We prove it by the contradiction. Suppose that $i_1$ (resp. $i_2$) occurs in at most $\delta+t-3$ support sets of partial check columns.
Then,  at most $\delta+t-3-t=\delta-3$ ones of these $\delta+t-3$  sets do not contain $i_2$ (resp. $i_1$).  Hence,
the two elements $i_1$ and $i_2$ must occur  in at least  two of any $\delta-1$ ones of these $\delta+t-3$ sets,
which contradicts Definition \ref{r_delta_i_one}.
\end{proof}

\begin{theorem}\label{thm_comb}
Let $\mathcal{C}$ be an $[n,k]$ systematic code. Then $\mathcal{C}$ has information $(r,\delta)_c$-locality if and only if there exists a $(k,R,1)$ packing
$\left([k],\mathcal{B}=\{B_j\}_{j=1}^{n_1}\right)$ with $\max(R)\leq r$ and $|\{B_j|i\in B_j\}|\ge \delta-1$ for any $1\le i\le k$.
\end{theorem}

\begin{proof}
Assume that $\mathcal{C}$ has information $(r,\delta)_c$-locality. Let $\mathcal{B}=\{B_j\}_{j=1}^{n_1}$ be the set obtained by Algorithm 1. It is easy to see that (i) $B_j\subseteq supp({\bf p}_j)$ for any $1\le j\le n_1$, which leads to $\max(R)=\max_{j=1}^{n_1}(|B_j|)|\le r$ due to $|supp({\bf p}_j)|\leq r$; (ii) Any two distinct elements occur simultaneously in at most one block. So, the two tuple $([k],\mathcal{B})$ is a packing by Definition \ref{Def_Pack}.

Note that  (i) Initially in Line 1, Algorithm 1,  $|\{B_j|i\in B_j\}|\ge \delta-1$ for any $1\le i\le k$ by Definition \ref{r_delta_i_one}; (ii) After the deletion in Line 3,  Algorithm 1, there are at least $\delta+t-2-(t-1)=\delta -1$ support sets of partial check columns including $i_1$ (resp. $i_2$) by Lemma \ref{lem_comb}. So, the output $\mathcal{B}=\{B_j\}_{j=1}^{n_1}$ of Algorithm 1 satisfies $|\{B_j|i\in B_j\}|\ge \delta-1$ for any $i \in [k]$.

The converse is also true from Definition \ref{r_delta_i_one}.
\begin{center}
\hrule
{\bf Algorithm 1:}{ Packing from the supports of partial check columns}
\vspace{1mm}
\hrule
\end{center}
\begin{center}
\begin{algorithmic}[1]
\INPUT{An $[n,k]$ systematic code  $\mathcal{C}$ with information $(r,\delta)_c$-locality.}
\STATE {Let $B_j=supp({\bf p}_j)$, $1\le j\le n_1$;}
\WHILE{there exist two distinct elements $i_1$, $i_2\in[k]$ satisfying
the pair  $(i_1, i_2)$ occurs in $t> 1$ ones of sets $P_1,\cdots,P_{n-k}$, say
$B_{i_1}, \cdots, B_{i_t}$;
      }
\STATE {Choose $t-1$ sets from $B_{i_1}, \cdots, B_{i_t}$ and delete one of $i_1$ and $i_2$ from each one;}
\ENDWHILE
\RETURN {$B_j$, $1\le j\le n_1$.}
\end{algorithmic}
\hrule
\end{center}
\end{proof}

\begin{example}
\label{exa2}
For $k=8$ and $r=\delta=3$, let
\begin{equation*}
\small
G=\left(
\begin{array}{rcccccccccccccl}
1&0&0&0&0&0&0&0&0&1&0&0&1&1\\
0&1&0&0&0&0&0&0&1&0&1&0&0&0\\
0&0&1&0&0&0&0&0&1&1&0&0&0&0\\
0&0&0&1&0&0&0&0&0&1&1&0&0&0\\
0&0&0&0&1&0&0&0&0&0&0&1&1&0\\
0&0&0&0&0&1&0&0&0&0&1&0&1&1\\
0&0&0&0&0&0&1&0&0&0&0&1&0&1\\
0&0&0&0&0&0&0&1&1&0&0&1&0&0\\
\end{array}
\right)
\end{equation*}
with
\begin{quote}
$supp({\bf p}_1)=\{2,3,8\}$, $supp({\bf p}_2)=\{1,2,3\},$ $supp({\bf p}_3)=\{2,4,6\},$ $supp({\bf p}_4)=\{5,7,8\},$ $supp({\bf p}_5)=\{1,5,6\}$, $supp({\bf p}_6)=\{1,6,7\}$
\end{quote}
We can check that the code $\mathcal{C}$ generated by $G$ is a $[14,8,3]$ code with information $(3,3)_c$-locality. Clearly $\{1,6\}$ occurs in $supp({\bf p}_5)$ and $supp({\bf p}_6)$. By Algorithm 1, we can delete one information symbol, for example $6\in supp({\bf p}_6)$, to get a packing as
$$\{\{2,3,8\},\{1,2,3\},\{2,4,6\},\{5,7,8\},\{1,5,6\},\{1,7\}\}$$
\end{example}

According to Theorem \ref{thm_comb}, a $(k,R,1)$ packing can be obtained by deleting some elements from
the supports of the partial check columns and further the new code still possesses information $(r,\delta)_c$-locality.
Specially, we can easily get the following conclusion by Lemma \ref{lem_comb} and Theorem \ref{thm_comb}.

\begin{corollary}\label{thm1-cor}
Let $\mathcal{C}$  be an $[n,k]$ systematic code. If at most one element of $[k]$ occurs in more than $\delta-1$ support sets of partial check columns, then
$([k],\{supp({\bf p}_j)\}_{j=1}^{n_1})$ naturally form a $(k,R,1)$ packing.
\end{corollary}

In general, we will show that at most $r-1$ elements need to be deleted in most cases for optimal systematic codes with information $(r,\delta)_c$-locality in the remainder of this section.

Given an $[n,k,d]$ systematic code $\mathcal{C}$ with information $(r,\delta)_c$-locality,
define
\begin{eqnarray*}
\delta_i=|\{j\in [n_1]|i\in supp({\bf p}_j)\}|
\end{eqnarray*}
for each element $i\in [k]$, i.e., the occurrence that an element appears in the support sets of partial check columns, and
\begin{eqnarray}\label{Eqn_def_Delta}
\Delta=\min_{1\leq i\leq k}\delta_i
\end{eqnarray}
By Definition \ref{r_delta_i_one}, $\delta_i\geq \delta-1$ for $i\in [k]$, thus
\begin{equation}\label{eqn_Delta}
\Delta\geq \delta-1
\end{equation}
For $j\in [n_1]$, let the Hamming weight of the partial check column ${\bf p}_{j}$  be $w_j$.
Then,
\begin{eqnarray}\label{Eqn_delta}
k\Delta \le\sum_{i=1}^{k}\delta_i=\sum_{j=1}^{n_1}w_j \leq n_1 r
\end{eqnarray}
since each element $i\in [k]$  appears in $\delta_i$ support sets of  partial check columns while  each set $supp({\bf p}_j)$ ($j\in [n_1]$) contains $w_j\le r$ elements, we get
\begin{eqnarray}\label{Eqn_n1}
n_1\ge \left\lceil {k\Delta\over r} \right\rceil
\end{eqnarray}

In what follows, we determine the exact value of $n_1$ for most cases about optimal systematic codes with information $(r,\delta)_c$-locality. We begin with two useful lemmas.

\begin{lemma}\label{Lem_Weight} With the notations as above, the Hamming weight of each row in matrix $P$ is no less than $d-1$.
\end{lemma}
\begin{proof}
The result directly follows from the fact that as a codeword, each row in $G=(I\,|\,P)$ has the Hamming weight no less than $d$.
\end{proof}

\begin{lemma}\label{Lem_Case2} For an $[n,k,d]_q$  systematic code with information $(r,\delta)_c$-locality, if $n_1=\lceil k\Delta/ r \rceil$
and $\delta\ge 4$, then $k> 2r$.
\end{lemma}
\begin{proof}
Assume that there are $m$ partial check columns ${\bf p}_{j_1},\cdots,{\bf p}_{j_m}$ respectively having the Hamming weight  $w_{j_1},w_{j_2},\cdots,$ $w_{j_m}$,
which are all less than $r$. Then, by \eqref{Eqn_delta}
\begin{eqnarray*}
k\Delta\le \sum_{j=1}^{n_1}w_j=(n_1-m)r+\sum_{t=1}^{m}w_{j_t}=\left\lceil{k\Delta \over r}\right\rceil r-\sum_{t=1}^{m}(r-w_{j_t})
\end{eqnarray*}
 Thus, we have
\begin{equation}\label{eqn_theorem_case_2_special}
 \sum_{t=1}^{m}(r-w_{j_t})\leq \left\lceil{k\Delta \over r}\right\rceil r-k\Delta\leq r-1
\end{equation}

Let $i_1$ and $i_2$ be two distinct integers in $[k]$. Firstly respectively consider the  $\delta-1$ repair groups of information symbols $c_{i_1}$ and $c_{i_2}$. Suppose that the Hamming weights of the corresponding
$2(\delta-1)$ partial check columns ${\bf p}_{1}',\cdots,{\bf p}_{\delta-1}'$ and ${\bf p}_{1}'',\cdots,{\bf p}_{\delta-1}''$  are $w_1',\cdots, w_{\delta-1}'$ and $w_1'',\cdots, w_{\delta-1}''$ respectively. Definition \ref{r_delta_i_one} tells us that each element $i\ne i_1$ (resp. $i\ne i_2$) appears at most once in the
support sets of partial check columns  $supp({\bf p}_{1}'),\cdots,supp({\bf p}_{\delta-1}')$ (resp. $supp({\bf p}_{1}''),\cdots,supp({\bf p}_{\delta-1}'')$), i.e.,
\begin{eqnarray*}
k\ge 1+ \sum_{t=1}^{\delta-1}(w_t'-1),~k\ge 1+ \sum_{t=1}^{\delta-1}(w_t''-1)
\end{eqnarray*}
and hence
\begin{eqnarray*}
2k &\ge & 2+ \sum_{t=1}^{\delta-1}(w_t'+w_t''-2)\nonumber\\
&=&  2+ 2(\delta-1)(r-1)-\sum_{t=1}^{\delta-1}(r-w_t')-\sum_{t=1}^{\delta-1}(r-w_t'')\label{Eqn_k_bound}
\end{eqnarray*}

Secondly, consider the $2(\delta-1)$ support sets of these partial check columns. Definition \ref{r_delta_i_one} implies that  $i_1$ and $i_2$ appear simultaneously in at most one of them, i.e.,
\begin{eqnarray*}
\rho=|\{{\bf p}_{1}',\cdots,{\bf p}_{\delta-1}'\}\cap \{{\bf p}_{1}'',\cdots,{\bf p}_{\delta-1}''\}|\le 1
\end{eqnarray*}
W.L.O.G., set ${\bf p}_{1}'={\bf p}_{1}''$ if $\rho=1$, which implies $w_1'=w_1''\ge 2$ because of  $i_1,i_2\in supp({\bf p}_{1}')$.

Thus by \eqref{eqn_theorem_case_2_special} we obtain
\begin{equation*}
  \sum_{t=1}^{\delta-1}(r-w_t')+\sum_{t=1}^{\delta-1}(r-w_t'')\le \sum_{t=1}^{m}(r-w_{j_t})+\rho(r-w_1'')<2(r-1)
\end{equation*}
 which gives
\begin{eqnarray*}
2k> 2+ 2(\delta-1)(r-1)- 2(r-1)=2+ 2(\delta-2)(r-1)
\end{eqnarray*}
i.e.,  $k> 2r$ if $\delta\geq 4$.
\end{proof}

\begin{theorem}\label{theorem_num_n_1}
For any optimal $[n,k,d]_q$  systematic code with information $(r,\delta)_c$-locality, $n_1= \left\lceil\frac{k(\delta-1)}r\right\rceil$
if one of the following conditions holds
\begin{itemize}
\item [C1.] $\delta\ge 4$;
\item [C2.] $\delta=3$, $k\ge 2r$, or ($k=r+\kappa$ and ($r/3\le \kappa<r/2$ or $2r/3\le \kappa<r$));
\item [C3.] $\delta=2$, $k\ge 2r$, or ($k=r+\kappa$ and $r/2\le \kappa<r$).
\end{itemize}
\end{theorem}

\begin{proof} Firstly,  $\Delta\geq \delta-1$ by \eqref{eqn_Delta}. Set $\Delta=\delta-1+l$ for integer $l\ge 0$.
Applying \eqref{eqn_bound_with_condition} in the place of $n=k+n_1+n_2$, we have
\begin{equation}\label{eqn_case_t}
 n_2= d +\left\lceil\frac{k(\delta-1)}{r}\right\rceil-\delta-n_1
\end{equation}
By \eqref{Eqn_def_Delta}, there exists an element $i_0\in[k]$ satisfying $\delta_{i_0}=\Delta$. That is,
 the weight of  row $i_0$ of the matrix $P$ is at most $\Delta+n_2$. It then follows from Lemma \ref{Lem_Weight} that
\begin{eqnarray}\label{eqn_delta}
\delta-1+l+n_2=\Delta+n_2\geq d-1
\end{eqnarray}

Applying \eqref{eqn_case_t} to \eqref{eqn_delta}, we get
\begin{equation}\label{eqn_case_delta1}
 n_1\le l+\left\lceil\frac{k(\delta-1)}{r}\right\rceil \le \left\lceil\frac{k(\delta-1+l)}{r}\right\rceil
\end{equation}
since $r\le k$.
On the other hand, by \eqref{Eqn_n1}, we have
\begin{equation}\label{eqn_case_delta2}
n_1\geq\left\lceil\frac{k(\delta-1+l)}{r}\right\rceil
\end{equation}
Combining \eqref{eqn_case_delta1} and \eqref{eqn_case_delta2}, we  then arrive at
\begin{equation}\label{eqn_case_delta3}
n_1=\left\lceil\frac{k(\delta-1+l)}{r}\right\rceil=l+\left\lceil\frac{k(\delta-1)}{r}\right\rceil
\end{equation}

Next we show that $l=0$ if C1 or C2 or C3 holds.  Otherwise if $l>0$,  by  \eqref{eqn_case_delta3}
\begin{eqnarray*}
\left\lceil\frac{k(\delta-1+l)}{r}\right\rceil -\left\lceil\frac{k(\delta-1)}{r}\right\rceil &=& l
\end{eqnarray*}
However, it is easily checked that the left hand side of the above equality  is larger than $l$ if (i) $k\ge 2r$; or (ii)  $\delta=3$, $k=r+\kappa$, ($r/3\le \kappa<r/2$ or $2r/3\le \kappa<r$)); (iii) $\delta=2$, $k=r+\kappa$, $r/2\le \kappa<r$. Recall from Lemma 6 that \eqref{eqn_case_delta3} and C1 (i.e., $\delta\ge 4$) lead to $k> 2r$. That is, there is always a contradiction for any one of C1, C2 and C3, which finishes the proof.
\end{proof}

Based on Theorem \ref{theorem_num_n_1}, we are able to get the following result.

\begin{theorem}\label{theorem_number_points}
Assume that C1 or C2 or C3 holds. Let $\mathcal{C}$ be an optimal $[n,k,d]_q$ systematic code with information $(r,\delta)_c$-locality.
If there exist $m$ elements $i_1,i_2,\cdots,i_m\in[k]$ such that $\delta_{i_t}>\delta-1$ for $1\leq t\leq m$, then
\begin{equation}\label{eqn_for_num_special_dots}
        m\leq \left\lceil\frac{k(\delta-1)}{r}\right\rceil r-k(\delta-1)
\end{equation}
\end{theorem}

\begin{proof} Note that $\delta_{i}>\delta -1$ if $i\in \{i_1,\cdots,i_m\}$ and $\delta_{i}=\delta -1$ otherwise. Thus,
we have
 \begin{equation*}
  m\leq \sum_{t=1}^{m}[\delta_{i_t}-(\delta-1)]=\sum_{i=1}^{k}[\delta_{i}-(\delta-1)]=\sum_{i=1}^{k}\delta_{i}-k(\delta-1)
  \leq \left\lceil\frac{k(\delta-1)}{r}\right\rceil r-k(\delta-1)
 \end{equation*}
 where the last inequality holds by applying Theorem \ref{theorem_num_n_1} to \eqref{Eqn_delta}.
\end{proof}

For example, we can check that the code in Example \ref{exa2} is an optimal $[14,8,3]$ code with information $(3,3)_c$-locality. From Theorem \ref{theorem_number_points} there are at most one pair of points that appears in more than one support sets of ${\bf p}_i$ for $1\leq i\leq n_1$. That is, $\{1,6\}$ occurs in $supp({\bf p}_5)$ and $supp({\bf p}_6)$.

By Theorem \ref{theorem_number_points}, we know that we only need to delete
 at most $\left\lceil k(\delta-1)/ r\right\rceil r-k(\delta-1)-1<r-1$
elements from the support sets of all the partial check columns to form a packing when C1 or C2 or C3 is satisfied.
Specifically,
in the following two cases, we do not need any deletion. In other words, the support sets of all
the partial check columns form a packing natively.

\begin{corollary}\label{corollary1}
Assume that C1 or C2 or C3 holds.  For any  $[n,k,d]_q$ optimal systematic code $\mathcal{C}$ with information $(r,\delta)_c$-locality,
the support sets of the partial check columns in $\mathcal{C}$ form a $(\delta-1)$-regular $(k,r,1)$ packing if $r|k(\delta-1)$.
\end{corollary}

\begin{proof}
In this case, $m=0$ in \eqref{eqn_for_num_special_dots}. Thus by Corollary \ref{thm1-cor}, the support sets of partial check columns naturally form a $(\delta-1)$-regular $(k,R,1)$ packing. Further by Theorem \ref{theorem_num_n_1}, $\sum\limits_{j=1}^{n_1} w_j=k(\delta-1)=n_1 r$ which gives $w_j= r$ for all $j\in[n_1]$ since $w_j\le r$ for all $j\in [n_1]$. Hence, the resultant packing is a $(\delta-1)$-regular $(k,r,1)$ packing.
\end{proof}

\begin{corollary}\label{corollary2}
Assume that C1 or C2 or C3 holds.  For any optimal systematic code with information $(r,\delta)_c$-locality,
then the support sets of the partial check columns in $\mathcal{C}$ form
\begin{itemize}
  \item a $(k,r,1)$ packing with $\lceil k(\delta-1)/r \rceil$ blocks; or
  \item a  $(\delta-1)$-regular $(k,\{r,r-1\},1)$  packing   having exactly one block of size $r-1$
\end{itemize}
provided that $k(\delta-1)\equiv r-1\pmod{r}$.
\end{corollary}

\begin{proof}
In this case, we have
\begin{eqnarray}\label{Eqn_r-1}
m=\sum\limits_{i=1}^k [\delta_{i}-(\delta-1)]\le\lceil k(\delta-1)/r \rceil r-k(\delta-1)=1
\end{eqnarray}
Thus by Corollary \ref{thm1-cor}, the support sets of partial check columns naturally form a $(k,R,1)$ packing.

Note that
\begin{equation*}
  \sum\limits_{i=1}^{n_1}(r-w_i)=n_1r-\sum\limits_{i=1}^{n_1}w_i=\lceil k(\delta-1)/r \rceil r-k(\delta-1)-\sum\limits_{i=1}^k [\delta_{i}-(\delta-1)]=1-m
\end{equation*}

For $m=1$, we have $\sum\limits_{i=1}^{n_1}(r-w_i)=0$ which gives $w_j=r$ for all $j\in [n_1]$. Thus the resultant packing is a $(k,r,1)$ packing with $\lceil k(\delta-1)/r \rceil$ blocks.

For $m=0$, we have $\sum\limits_{i=1}^{n_1}(r-w_i)=1$, thus there exists one element $j_0\in[n_1]$ such that $w_{j_0}=r-1$ and $w_j=r$ for $j\in[n_1]\setminus\{j_0\}$. The resultant packing is a $(\delta-1)$-regular $(k,\{r,r-1\},1)$  packing with exactly one block of size $r-1$, since $m=0$ means $\delta_i=\delta-1$ for $i\in [n_1]$.
\end{proof}

\section{Optimal locally repairable codes from Packing}\label{sec-construction}

By  Theorem \ref{theorem_num_n_1}, $n_1=\lceil k(\delta-1)/r \rceil$ for most optimal $(r,\delta)_c$-locally repairable codes. This implies that the $n_1$ support sets of the partial check columns contain  at least $k(\delta-1)$ elements. Further by Theorem \ref{theorem_number_points} and Algorithm 1, we can obtain packing by deleting at most $r-1$ elements. Since $r-1$ is relatively small compared with $k(\delta -1)$, it is naturally to ask whether packing can be used to construct locally repairable codes conversely. In this section, we answer this issue in two cases $n_2=0$ and $n_2>0$, respectively.

\subsection{The case $n_2=0$}

In this subsection, we assume that $n_2=0$.

{\bf Construction A:} For any positive integers $k$ and $r$, if there exists a $(k,R,1)$ packing,
$(X,\mathcal{B})$ with $\mathcal{B}=\{B_1,B_2,\cdots,B_{n_1}\}$, then a code $\mathcal{C}$ can be generated
 by the following $k\times(k+n_1)$ matrix
\begin{equation}\label{eqn_def_generator_m_A}
G=\left({\bf e}_1,{\bf e}_2,\ldots,{\bf e}_{k}\ |\ {\bf p}_1,{\bf p}_2,\cdots,{\bf p}_{n_1}\right),
\end{equation}
where ${\bf p}_i=(p^i_1,p^i_2,\cdots p^i_k)^\top$ is the $k$-dimensional vector defined as
\begin{eqnarray*}
p^i_j=\left\{
\begin{array}{ll}
1,\,\,\,\,&{\rm if}\,\,j\in B_i\\
0,\,\,\,\,&{\rm otherwise}
\end{array}
\right.
\end{eqnarray*}
for $1\leq i\leq n_1$ and $\top$ is the transpose operator.

\begin{remark} (i) Since $G$ in \eqref{eqn_def_generator_m_A} is a binary matrix, the resultant code $\mathcal{C}$ can be as simple as a binary code.

(ii) In \cite{PHO},  codes with information $(r,\delta)_c$-locality were constructed via
partial geometry. In fact, partial geometry is a special case of packing but with very strict restriction of parameters so that
only few instances are known till now \cite{CD}. In this sense, Construction A is a generalization of the one in \cite{PHO}.
\end{remark}

From Lemma \ref{new_bound} and Theorem \ref{thm_comb}, the following result can be obtained.

\begin{theorem}\label{theorem_cons_A}
$\mathcal{C}$ generated in Construction A is a $[k+n_1,k,d]_q$ systematic code with $d=\delta$ and information $(r,\delta)_c$-locality where  $\delta-1=\min\limits_{i\in X}|\{j\ |\ i\in B_j\}|$ and $r=\max{R}$. Further,
\begin{itemize}
     \item $\mathcal{C}$ is an optimal systematic code with information $(r,\delta)_c$-locality if and only if $n_1=\left\lceil k(\delta-1)/ r\right\rceil$;
     \item $\mathcal{C}$ has the optimal update-efficiency if $(X,\mathcal{B})$ is a $(\delta-1)$-regular $(k,R,1)$ packing.
\end{itemize}
\end{theorem}

\begin{proof} Firstly, $\mathcal{C}$ is a $[k+n_1,k,d]_q$ systematic code with information $(r,\delta)_c$-locally by Theorem \ref{thm_comb}.

Secondly, we prove that $d=\delta$. It is known from Construction A that there exists a row of $G$ with the Hamming weight $\delta$. That is, there is a codeword  in $\mathcal{C}$
with the Hamming weight $\delta$, which implies $d\leq \delta$. Therefore,  it is sufficient to show that $\mathcal{C}$ can tolerate any $\delta-1$ symbol erasures.
Let $D$ denote the set of all the erasure information symbols. For any information symbol $c_i\in D$,  note that
there are at least $\delta-1$ repair groups, say $R^{(i)}_1,R^{(i)}_2,\cdots,R^{(i)}_{\delta-1}$. Then by the Pigeonhole Principle, there must exist an integer $j\in [\delta-1]$ such that $(D\setminus\{c_i\})\bigcap R^{(i)}_j=\emptyset$ since  $|D\setminus \{c_i\}|\leq \delta -2$, which implies that repair group $ R^{(i)}_j$ can be used to repair the erasure information symbol $c_i$.  As for the erasure check symbols, they can be subsequently repaired by all the information symbols.

Then, by Lemma \ref{new_bound}, we conclude that $\mathcal{C}$ is optimal if and only if
\begin{equation*}
\delta= d = n_1+k-k-\left\lceil\frac{k(\delta-1)}{r}\right\rceil +\delta=n_1-\left\lceil\frac{k(\delta-1)}{r}\right\rceil+\delta
\end{equation*}
i.e., $n_1=\left\lceil\frac{k(\delta-1)}{r}\right\rceil$.

Finally, if the packing $(X,\mathcal{B})$ is a $(\delta-1)$-regular $(k,R,1)$, then $\mathcal{C}$ has update-efficiency $t=\delta=d$ by
Definition \ref{Def_Update_New_1} since each row of the generator matrix $G$ in \eqref{eqn_def_generator_m_A} built on the $(\delta-1)$-regular packing
has the Hamming weight $\delta$, which is optimal due to Definition \ref{Def_Update_New}.
\end{proof}

Combining  Theorem \ref{theorem_cons_A} with Corollaries \ref{corollary1} and \ref{corollary2}, we immediately have the following
sufficient and necessary conditions for the cases  $r|k(\delta-1)$ and $k(\delta-1)\equiv r-1\pmod{r}$.

\begin{corollary}\label{corollary_n_2=0_r_mid_k}
Assume that C1 or C2 or C3 holds. When $r|k(\delta-1)$ and $n_2=0$,  the systematic code has the optimal information $(r,\delta)_c$-locality and the optimal update-efficiency  if and only if the support sets of all the partial check symbols form a $(\delta-1)$-regular $(k,r,1)$ packing.
\end{corollary}

\begin{corollary}\label{theorem_Suf_Nece_condition}
Assume that C1 or C2 or C3 holds. When $k(\delta-1)\equiv r-1\pmod{r}$ and $n_2=0$, the systematic code has the  optimal information $(r,\delta)_c$-locality if and only if
 the support sets of the partial check symbol form
\begin{itemize}
  \item a $(k,r,1)$ packing with $\lceil k(\delta-1)/r \rceil$ blocks; or
  \item a  $(\delta-1)$-regular $(k,\{r,r-1\},1)$  packing   having exactly one block of size $r-1$
\end{itemize}
where the code corresponding to $(\delta-1)$-regular has the optimal update-efficiency as well.
\end{corollary}

\begin{example}
The two tuple $(X,\mathcal{B})$ with $X=[8]$ and
$$\mathcal{B}=\{\{2,3,8\},\,\{3,4,1\},\,\{4,5,2\},\,\{5,6,3\},\,\{6,7,4\},\,\{7,8,5\},\,\{8,1,6\},\,\{1,2,7\}\}$$
is a $3$-regular $(8,3,1)$ packings. Then, the generator matrix $G$ in \eqref{eqn_def_generator_m_A} is just the  one in \eqref{eqn_matrix}, which gives
an optimal $[16,8,4]$ systematic code with information $(3,4)_c$-locality and the optimal update-efficiency.
\end{example}

In the rest of this subsection, we apply packings with $n_1=\left\lceil k(\delta-1)/r\right\rceil$ to get  some optimal locally repairable codes in two cases by Construction A.

Case 1. $r|k(\delta-1)$. In this case, the construction of optimal locally repairable codes is equivalent to finding
regular packings with parameters $k,\delta$ satisfying C1 or C2 or C3 by Corollary \ref{corollary_n_2=0_r_mid_k}. This is to say, we only need to consider regular packings. In the literature, there are many known regular packings \cite{CK,CY,CSW,MOKL}.
As an illustration, for any prime power $q$, we list some regular packings with flexible block size  and the resultant optimal locally repairable codes in Table I.

\begin{table}[h]

\center
\caption{ \label{table1} Some known regular packings and new optimal locally repairable codes for $r|k(\delta-1)$}
 \begin{tabular}{|c|c|c|c|c|}
\hline
Parameters of local repairable & Parameters of & \multirow{1}{*}{Information locality}&\multirow{2}{*}{Constraints}& \multirow{2}{*}{References}\\
codes $[n,\,k,\,d]$&$(\delta-1)$-regular packing&$(r,\delta)_c$&&\\
\hline
\hline
$\left[(t+1)\frac{q^3-1}{q-1},\frac{q^3-1}{q-1},t(q+1)+1\right]$&$t(q+1)$-$\left(\frac{q^3-1}{q-1},q+1,1\right)$&$(q+1,t(q+1)+1)_c$&$t=1$& {\cite{CSW}}\\
\hline
\multirow{2}{*}{$\left[(t+1)\frac{q^{x+1}-1}{q-1},\frac{q^{x+1}-1}{q-1},t(q+1)+1\right]$}
&\multirow{2}{*}{$t(q+1)$-$\left(\frac{q^{x+1}-1}{q-1},q+1,1\right)$}&\multirow{2}{*}{$(q+1,t(q+1)+1)_c$}&$x$ is even and $1\le t\le\frac{q^{x}-1}{q^2-1}$&\multirow{2}{*}{\cite{CSW}}\\
\cline{4-4}
&&&$x$ is odd and $1\le t\le\frac{q^{x}-1}{q^2-q}$&\\
\hline
$[(t+1)q,q,rt+1]$&\multirow{1}{*}{$rt$-$(q,r,1)$}&$(r,rt+1)_c$&$q=r(r-1)x+1$ and $1\le t\le r$ & \multirow{1}{*}{\cite{CK}}\\
\hline
$[(t+1)(q^x-1),q^x-1,qt+1]$&\multirow{1}{*}{$qt$-$(q^x-1,q,1)$}&$(q,qt+1)_c$&$x\geq 2$
and $1\le t \le\frac{q^{x-1}-1}{q-1}$& \multirow{1}{*}{\cite{CK,MOKL}}\\
\hline
\multirow{2}{*}{$[(t+1)rq,rq,rt+1]$}&\multirow{2}{*}{$rt$-$(rq,r,1)$}&\multirow{2}{*}{$(r,rt+1)_c$}&$r|(q-1)$, $(r-1)^2>q-1$ & \multirow{2}{*}{\cite{CY}}\\
&&&and $1\le t\le\frac{q-1}{r}$&\\
\hline
\end{tabular}
\end{table}

Case 2. $r\nmid k(\delta-1)$. In the combinatorial design theory, most packings with $r\nmid k(\delta-1)$ can be constructed by regular $(k,r,1)$ packings directly. So we only list some known results of these packings hereafter to construct the optimal repairable codes.

For the case $k(\delta-1)\equiv r-1\pmod{r}$, there exists a $(p^2+1,p,1)$ packing with $\left\lceil\frac{(p-1)(p^2+1)}{p}\right\rceil$ blocks for any prime number $p$, where $p^2$ elements occur $p-1$ times and one point occurs $p$ times. Let $k=p^2+1$, $\delta=p$, and $r=p$, then by Construction A we can obtain  an optimal  $[2p^2-p+2,p^2+1,p]$ locally repairable code with information $(p,p)_c$-locality.
Further, assume that point $i\in [p^2+1]$ occurs $p$ times,  let $([p^2+1],\mathcal{B}')$ be the $(p^2+1,\{p,p-1\},1)$ packing obtained
by deleting $i$ from one block of $\mathcal{B}$. Then, by Construction A, we can get a $[2p^2-p+2,p^2+1,p]$ locally repairable code with the optimal information $(p,p)_c$-locality and the optimal update-efficiency.

For the case $r\nmid(\delta-1)k$ and $(\delta-1)k\not\equiv r-1\pmod{r}$, there exists a $(q^2-1,q,1)$ packing with $q^2-q$ blocks for any prime power $q$, where  $q-1$ elements appear in $q$ blocks and $q^2-q$ elements appear in $q-1$ blocks. Let $k=q^2-1$, $\delta=q$ and $r=q$, then we can generate an optimal $[2q^2-q,q^2-1,q]$ locally repairable code with information $(q,q)$-locality by Construction A.

\subsection{The case $n_2>0$}\label{sec-construction-packing}
When $n_2>0$, we construct the optimal systematic code with information $(r,\delta)_c$-locality based on resolvable packing and MDS code as follows.

{\bf Construction B:} Let $\mathcal{W}$ be an $[n,k,d]_q$ MDS code with generator matrix
\begin{eqnarray}\label{Eqn_codeW}
G=({\bf e}_1,{\bf e}_2,\cdots,{\bf e}_k\ |\ {\bf p}_1,{\bf p}_2,\cdots,{\bf p}_{n-k})
\end{eqnarray}
where the column ${\bf p}_{i}$ is denoted by  ${\bf p}_{i}=(p_{i,1},p_{i,2},\ldots,p_{i,k})^\top$.
Let $([k],\mathcal{B})$ be a $(k,R,1;u)$ resolvable packing with $\mathcal{B}=\bigcup\limits_{1\leq i\leq u}\mathcal{B}_i$,
$\mathcal{B}_i=\{B_{i,1},\cdots, B_{i,|\mathcal{B}_i|}\}$, and $u<n-k$. For any block $B_{i,j}$, $i\in [u]$ and  $j\in [|\mathcal{B}_i|]$,  define ${\bf p}_{i}^{B_{i,j}}=(p_{i,1}^{B_{i,j}},\ldots$, $p_{i,k}^{B_{i,j}})^\top$ as
\begin{equation}\label{eqn_local_check_def}
p_{i,l}^{B_{i,j}}=\left\{
\begin{array}{ll}
p_{i,l},\ \ & l\in B_{i,j}\\
0,\ \ & l\in [k] ~\mathrm{and}~ l\notin B_{i,j}
\end{array}
\right.
\end{equation}
Then, a new code $\mathcal{C}$ can be generated by the following generator matrix
\begin{equation}\label{eqn_constructionB}
\begin{split}
G(\mathcal{W},\mathcal{B})=\left({\bf e}_1,{\bf e}_2,\cdots,{\bf e}_k\ \left|\right.\ {\bf p}_1^{B_{1,1}},\ldots,{\bf p}_1^{B_{1,|\mathcal{B}_1|}},{\bf p}_2^{B_{2,1}},\ldots,
{\bf p}_u^{B_{u,1}},\ldots,{\bf p}_u^{B_{u,|\mathcal{B}_u|}},{\bf p}_{u+1},\ldots,{\bf p}_{n-k}\right)
\end{split}
\end{equation}

\begin{remark}
When $u=1$,  the code generalized by construction B based on the $(k,R=\{r\},1;1)$ resolvable packing is exactly the Pyramid Code \cite{HCL}.
In this sense, Construction B is a generalization of the Pyramid Code.
\end{remark}

\begin{theorem}\label{theorem_cons_C}
The code $\mathcal{C}$ generated in Construction B is a $[n',k,d']_q$ systematic code with information $(r,\delta)_c$-locality where $n'=n+\sum_{i\in [u]}(|\mathcal{B}_i|-1)$, $d'= n-k+1$, $r=\max{R}$, and $\delta=u+1$. Moreover,
\begin{itemize}
     \item $\mathcal{C}$ is an optimal systematic code with information $(r,\delta)_c$-locality if $n_1=\left\lceil\frac{k(\delta-1)}{r}\right\rceil$, i.e., the resolvable packing has $\left\lceil\frac{k(\delta-1)}{r}\right\rceil$ blocks;
     \item $\mathcal{C}$ has the optimal update-efficiency.
\end{itemize}
\end{theorem}
\begin{proof}
Firstly we show that $d'= n-k+1$.  As the generator matrix with canonical  form of a systematic $(n,k)$ MDS code, $G$ in \eqref{Eqn_codeW} has the Hamming weight $n-k+1$ for each row and
then the Hamming weight $k$ for each check column.
Observe from \eqref{eqn_local_check_def}, we have the fact that
 each column ${\bf p}_{l}$ ($1\le l\le u$) in $G$
is extended to a $k\times |\mathcal{B}_l|$ sub-matrix  $({\bf p}_l^{B_{l,1}},\ldots,{\bf p}_l^{B_{l,|\mathcal{B}_l|}})$ in $G(\mathcal{W},\mathcal{B})$ with exactly
one nonzero entry in each row
 since $\mathcal{B}_l=\{B_{l,1},\cdots,B_{l,|B_l|}\}$ is a partition of $[k]$. This fact implies
that each row in $G(\mathcal{W},\mathcal{B})$  has the Hamming weight $u+(n-k-u)+1=n-k+1$ too. That is, there are codewords of $\mathcal{C}$  with the Hamming weight $n-k+1$.
Thus, we have $d'\le n-k+1$. On the other hand,  for any given $k$ information symbols $m_1,\cdots, m_k$, let $v=(m_1,\cdots,m_k$, $v_{1},\cdots,v_{n-k})$  and $c=(m_1,\cdots,m_k, c_{1,1},\cdots$,  $c_{1,|\mathcal{B}_1|},\cdots,  c_{u,1},\cdots,c_{u,|\mathcal{B}_u|},c_{u+1}, \cdots, c_{n-k})$ be the codeword generated by \eqref{Eqn_codeW} and \eqref{eqn_constructionB} respectively.
Then, the fact clearly indicates that
\begin{eqnarray*}
v_i=\sum_{j=1}^{|\mathcal{B}_i|} c_{i,j},~ 1\le i\le u
\end{eqnarray*}
which implies that terms $c_{i,j}$ in right hand side are not all zeros unless $v_i=0$. Noting that $v_i=c_i,~ u<i\le n-k$, we then have that the Hamming weight of $c$  is always no less than that of $v$. This is to say, $d'\geq n-k+1$. So, we get $d'=n-k+1$.

Secondly, given $i\in[k]$, there must exists a block in $\mathcal{B}_l$ containing $i$, denoted by $B_{l,i_l}$, since $\mathcal{B}_l=\{B_{l,1},\cdots,B_{l,|B_l|}\}$ is a partition of $[k]$. It is easily
seen that the partial check columns ${\bf p}_{1,i_1},\cdots,{\bf p}_{u,i_u}$ form the repair group for the systematic symbol $i$.
Therefore,  $\mathcal{C}$ generated in Construction B is a $[n',k,d']_q$ systematic code with information $(r,\delta)_c$-locality with
$r=\max{R}$ and $\delta=u+1$.

Thirdly, if the resolvable packing has $\left\lceil k(\delta-1)/ r\right\rceil$ blocks, i.e., $n_1=\sum_{i\in [u]}|\mathcal{B}_i|=\left\lceil k(\delta-1)/ r\right\rceil$,
then $n'=n+\sum_{i\in [u]}(|\mathcal{B}_i|-1)=n+\left\lceil k(\delta-1)/ r\right\rceil-u$. Thus, we have
\begin{eqnarray*}
n'-k-\left\lceil\frac{k(\delta-1)}{r}\right\rceil+\delta&=&n-k+1\\
&=&d\\
&=&d'
\end{eqnarray*}
where the second identity comes from the MDS property of the code $\mathcal{W}$. Then, the minimal Hamming distance
$d'$ achieves the lower bound in \eqref{eqn_bound_with_condition}. That is, the code $\mathcal{C}$ is is an optimal systematic code with information $(r,\delta)_c$-locality.

Finally, the optimal update-efficiency of the code  $\mathcal{C}$ follows from Definition \ref{Def_Update_New_1} because
all the  rows of the generator matrix in \eqref{eqn_constructionB} have the Hamming weight $d'=n-k+1$ as stated above.
\end{proof}

\begin{example}
Let $q=2^8$ and $k=8$. Then
\begin{equation*}
\small
W=\left(
\begin{array}{rcccccccccccccccl}
1&0&0&0&0&0&0&0&1&1&1&1&1&1&1&1\\
0&1&0&0&0&0&0&0&\alpha_1&\alpha_2&\alpha_3&\alpha_4&\alpha_5&\alpha_6&\alpha_7&1\\
0&0&1&0&0&0&0&0&\alpha_1^2&\alpha_2^2&\alpha_3^2&\alpha_4^2&\alpha_5^2&\alpha_6^2&\alpha_7^2&1\\
0&0&0&1&0&0&0&0&\alpha_1^{2^2}&\alpha_2^{2^2}&\alpha_3^{2^2}&\alpha_4^{2^2}&\alpha_5^{2^2}&\alpha_6^{2^2}&\alpha_7^{2^2}&1\\
0&0&0&0&1&0&0&0&\alpha_1^{2^3}&\alpha_2^{2^3}&\alpha_3^{2^3}&\alpha_4^{2^3}&\alpha_5^{2^3}&\alpha_6^{2^3}&\alpha_7^{2^3}&1\\
0&0&0&0&0&1&0&0&\alpha_1^{2^4}&\alpha_2^{2^4}&\alpha_3^{2^4}&\alpha_4^{2^4}&\alpha_5^{2^4}&\alpha_6^{2^4}&\alpha_7^{2^4}&1\\
0&0&0&0&0&0&1&0&\alpha_1^{2^5}&\alpha_2^{2^5}&\alpha_3^{2^5}&\alpha_4^{2^5}&\alpha_5^{2^5}&\alpha_6^{2^5}&\alpha_7^{2^5}&1\\
0&0&0&0&0&0&0&1&\alpha_1^{2^6}&\alpha_2^{2^6}&\alpha_3^{2^6}&\alpha_4^{2^6}&\alpha_5^{2^6}&\alpha_6^{2^6}&\alpha_7^{2^6}&1\\
\end{array}
\right)
\end{equation*}
is a generator matrix of a $[16,8,9]_{q}$ MDS code, where $\alpha_i=\beta^i$ for $1\leq i\leq 7$ and $\beta$ is a primitive element of $\mathbb{F}_{2^8}$. Clearly the two tuple $(X, \mathcal{B})$ with $X=[8]$ and $\mathcal{B}=\mathcal{B}_1\bigcup \mathcal{B}_2$ is a $(8,\{3,2\},1;2)$ resolvable packing, where
\begin{equation*}
\begin{split}
\mathcal{B}_1=\{\{2,3,8\},\,\{6,7,4\},\,\{1,5\}\},\ \ \mathcal{B}_2=\{\{3,4,1\},\,\{7,8,5\},\,\{2,6\}\}\\
\end{split}
\end{equation*}
Then from Construction B, we have a $[20,8,9]_{q}$ code generated by  matrix
\begin{equation*}
\small
G=\left(
\begin{array}{rcccccccccccccccccccl}
1&0&0&0&0&0&0&0&0&0&1&1&0&0&1&1&1&1&1&1\\
0&1&0&0&0&0&0&0&\alpha_1&0&0&0&0&\alpha_2&\alpha_3&\alpha_4&\alpha_5&\alpha_6&\alpha_7&1\\
0&0&1&0&0&0&0&0&\alpha_1^2&0&0&\alpha_2^2&0&0&\alpha_3^2&\alpha_4^2&\alpha_5^2&\alpha_6^2&\alpha_7^2&1\\
0&0&0&1&0&0&0&0&0&\alpha_1^{2^2}&0&\alpha_2^{2^2}&0&0&\alpha_3^{2^2}&\alpha_4^{2^2}&\alpha_5^{2^2}&\alpha_6^{2^2}&\alpha_7^{2^2}&1\\
0&0&0&0&1&0&0&0&0&0&\alpha_1^{2^3}&0&\alpha_2^{2^3}&0&\alpha_3^{2^3}&\alpha_4^{2^3}&\alpha_5^{2^3}&\alpha_6^{2^3}&\alpha_7^{2^3}&1\\
0&0&0&0&0&1&0&0&0&\alpha_1^{2^4}&0&0&0&\alpha_2^{2^4}&\alpha_3^{2^4}&\alpha_4^{2^4}&\alpha_5^{2^4}&\alpha_6^{2^4}&\alpha_7^{2^4}&1\\
0&0&0&0&0&0&1&0&0&\alpha_1^{2^5}&0&0&\alpha_2^{2^5}&0&\alpha_3^{2^5}&\alpha_4^{2^5}&\alpha_5^{2^5}&\alpha_6^{2^5}&\alpha_7^{2^5}&1\\
0&0&0&0&0&0&0&1&\alpha_1^{2^6}&0&0&0&\alpha_2^{2^6}&0&\alpha_3^{2^6}&\alpha_4^{2^6}&\alpha_5^{2^6}&\alpha_6^{2^6}&\alpha_7^{2^6}&1\\
\end{array}
\right)
\end{equation*}
 It is easy to check that it has the optimal information $(3,3)_c$-locality and the optimal update-efficiency.
\end{example}

In Combinatorics, the resolvable packings with the maximal number of blocks are the main concern. For example in the classical book \cite{CD},  only the resolvable packings with the block length $3$ and $4$ were discussed, which gives locally repairable code with parameter $r=3,4$ by  Theorem \ref{theorem_cons_C}. However, for our purpose, the resolvable packings with various parameter $r$ are also desirable since any resolvable packing with  $\left\lceil\frac{k(\delta-1)}{r}\right\rceil$ blocks can be used to construct optimal locally repairable code. Therefore, we present a construction of resolvable packing based on difference matrix as follows.

\begin{lemma}\label{cdm_oa}
For any positive integers $r$, $u$ and $k$, let $D=(d_{j,i})$, $j\in[r]$ and $i\in [u]$, be a $r\times u$ matrix over an additive group $\mathbb{M}$ of order $k$.
Define $\mathcal{B}_i=\{D_i+ a\ |\  a\in \mathbb{M}\}$ for $1\le i\le u$ where $D_i+ a=\{(j,d_{j,i}+ a )\ |\ j\in [r]\}$.
If for any $1 \le s \ne t \le r$, the $u$  differences $d_{s,i}- d_{t,i}$, $1 \le i \le u$, are distinct  over $\mathbb{M}$, then
$([r]\times \mathbb{M} ,\cup_{i\in [u]}\mathcal{B}_i)$  is  a $(rk,r,1;u)$ resolvable packing.
\end{lemma}
\begin{proof}
First it is easy to check that $\mathcal{B}_i$ is a partition of $[r]\times \mathbb{M}$ for each $1\le i\le u$ since $\mathbb{M}$ is an additive group.
Next let us show that $\mathcal{B}=\bigcup_{i=1}^{u}\mathcal{B}_i$ forms a packing over $[r]\times \mathbb{M}$. For any two distinct integers $i_1, i_2\in [u]$ and any two elements $a_1, a_2\in \mathbb{M}$, let $B_{i_1, a_1}=D_{i_1}+ a_1\in \mathcal{B}_{i_1}$ and $B_{i_2, a_2}=D_{i_2}+ a_2\in \mathcal{B}_{i_2}$. We only need to show $|B_{i_1, a_1}\bigcap B_{i_2, a_2}|\leq 1$ always holds. Suppose to the contrary that $|B_{i_1, a_1}\bigcap B_{i_2, a_2}|\geq 2$, say $x\neq y\in B_{i_1, a_1}\bigcap B_{i_2, a_2}$. Then there must exist two distinct integers $j_1$, $j_2\in[r]$ such that
\begin{eqnarray*}
x&=&(j_1, d_{j_1, i_1}+ a_1)=(j_1, d_{j_1, i_2}+ a_2)\\
y&=&(j_2, d_{j_2, i_1}+ a_1)=(j_2, d_{j_2, i_2}+ a_2)
\end{eqnarray*}
Then we have
\begin{eqnarray*}
d_{j_1, i_1}- d_{j_2, i_1}=d_{j_1, i_2}- d_{j_2, i_2}
\end{eqnarray*}
This is a contradiction to our assumption since $i_1 \neq i_2$. So $([r]\times \mathbb{M} ,\mathcal{B})$ is a packing, which completes the proof.
\end{proof}

\begin{definition} \rm
\label{d311}
For any two positive integers $k$, $r$, let $\mathbb{M}$ be an additive group of order $k$. A difference matrix $(k,r,1)$ DM is a $r \times k$ matrix $D = (d_{j, i})$ with $d_{j, i} \in \mathbb{M}$
such that for any $1 \le s \ne t \le r$, the differences $d_{s,i}- d_{t,i}$ over $\mathbb{M}$, $1 \le i \le k$, comprise all the elements of $\mathbb{M}$.
\end{definition}

When $k$ is a prime power, there always exists a $(k,r,1)$ DM over $\mathbb{F}_k$ for any $1\leq r\leq k$ \cite{CD}. If we delete any $k-u$ columns from the DM for $1\leq u\leq k$,  and apply the new matrix to Lemma \ref{cdm_oa}, we have the following result.

\begin{corollary}\label{cor-lem1}
For any prime power $k$, a resolvable $(rk,r,1;u)$ packing always exists for all positive integers $1\leq r,u\leq k$.
\end{corollary}

By Construction B and Corollary \ref{cor-lem1}, the following result can be obtained immediately.
 \begin{corollary}
For any prime power $k$, if there exists an $[n,k,d]_q$ MDS code, there exists an optimal systematic $[n-u+uk,k,d]_q$ code with  information $(r,u+1)_c$-locality and the optimal update-efficiency, where $r,u\in[k]$ and $n-rk\geq u$.
\end{corollary}

\begin{remark}
(i) Resolvable designs \cite{RPDV} and partial geometry \cite{PHO} are also introduced to construct optimal locally repairable codes with information $(r,\delta)_c$-locality. Both resolvable designs and partial geometry are special cases of packings, thus Construction B can be seen as a generalization of those constructions in \cite{RPDV} and \cite{PHO}. Notably, our construction can yield optimal locally repairable codes with new parameters compared with the known ones, since
\begin{itemize}
  \item {resolvable designs only work in the case $r|k$;}
  \item {partial geometry contains many restrictions on its possible parameters for instance $r|k(\delta-1)$, $(\delta-1)|r|\mathcal{B}|$ and so on (for more details the read can refer to \cite{CD} for definition of partial geometry). }
\end{itemize}

(ii) Corollaries \ref{corollary1} and \ref{corollary2} hint that packings are also necessary condition for optimal locally repairable codes with information $(r,\delta)_c$-locality for the cases $k(\delta-1)\pmod r \in \{0,1\}$.

(iii) From Theorems \ref{theorem_cons_A} and \ref{theorem_cons_C}, the rates of optimal locally repairable codes based on packings can be given as $\frac{k}{k+n_1}=\frac{k}{k+\lceil\frac{k(\delta-1)}{r}\rceil}$ and $\frac{k}{n+n_1-(\delta-1)}=\frac{k}{n+(\frac{k}{r}-1)(\delta -1)}$ respectively. In this sense, if $\mathcal{C}$ has
large $\delta$, i.e., high parallel reading ability then the code rate is low. That is, we can
choose suitable $\delta$ to tradeoff between parallel reading ability and code rate.
\end{remark}

\section{Conclusion}\label{sec-conclusion}
 In this paper, we first gave characterization of locally repairable codes from combinatorial design theory, which establishes a close relationship between optimal
locally repairable codes and packings. Next, we showed that regular packings and resolvable packings can be used to construct optimal locally repairable codes.
In particular, Constructions A and B were proposed.

By Constructions A and B, it is known that  packings  and resolvable packings with $\left\lceil k(\delta-1)/r\right\rceil$ blocks play important roles in generating optimal
locally repairable codes. Then, if more packings  and resolvable packings with $\left\lceil k(\delta-1)/r\right\rceil$ blocks can be constructed then more optimal repairable codes can be yielded. Thus, the reader is invited to construct these kinds of packings.

\end{document}